\documentclass[10pt]{article}

\usepackage{amsmath,latexsym,graphicx,color,amsfonts,amsthm}
\usepackage{fullpage}

\renewcommand{\epsilon}{\varepsilon}
\newcommand{\remove}[1]{}
\newcommand{\ftwo}{{\cal F}^{(2)}}

\newtheorem{theorem}{Theorem}
\newtheorem{proposition}[theorem]{Proposition}
\newtheorem{lemma}[theorem]{Lemma}

\newtheorem{claim}[theorem]{Claim}

\newcommand{\opt}{\mbox{\sc opt}}

\newcommand{\vett}{\mathbf }
\renewcommand{\b}{\vett b}
\newcommand{\bi}{\vett b_{-i}}

\newcommand{\prob}[1]{\mbox{\rm Prob}\left[\,#1\,\right]}
\newcommand{\expec}[1]{\mbox{\rm E}\left[\,#1\,\right]}

\begin{document}

\title{Using Lotteries to Approximate the Optimal Revenue
\thanks{Supported by EPSRC grants EP/G069239/1 and EP/G069034/1 ``Efficient Decentralised Approaches in Algorithmic
Game Theory.''}}

\author{Paul W. Goldberg\thanks{
Dept of Computer Science, University of Liverpool, Ashton Street,  Liverpool L69 3BX, U.K. E-mail: {\tt P.W.Goldberg@liverpool.ac.uk}} \and Carmine Ventre\thanks{School of Computing, Teesside University, Middlesbrough, TS1 3BA, U.K. E-mail: {\tt C.Ventre@tees.ac.uk}}}

\date{}

\maketitle

\begin{abstract}
There has been much recent work on the revenue-raising properties of truthful mechanisms for selling goods to selfish bidders. Typically the revenue of a mechanism is compared against a benchmark
(such as, the maximum revenue obtainable by an omniscient seller selling at a fixed price to at least two customers), with a view to understanding how much lower the mechanism's revenue is than the benchmark, in the worst case.
We study this issue in the context of {\em lotteries}, where the seller may sell a probability
of winning an item. We are interested in two general issues.
Firstly, we aim at using the true optimum revenue as benchmark for our auctions.
Secondly, we study the extent to which the expressive power resulting from lotteries, helps to improve the worst-case ratio. 
We study this in the well-known context of {\em digital goods}, where the production cost is zero.
We show that in this scenario, collusion-resistant lotteries (these are lotteries for which no coalition
of bidders exchanging side payments has an advantage in lying) are as powerful as truthful ones. 
\end{abstract}

\section{Introduction}
We consider a scenario in which a ``digital good'' is to be sold to many potential
buyers, with the objective of maximizing the revenue. A digital good is assumed to be provided with unlimited supply and to have no
cost of production. Given a set of selfish buyers who may have diverse valuations for the good,
a theoretical optimum for the revenue (commonly denoted \opt) is given by the sum of the buyers'
valuations.

In a standard mechanism, the allocation algorithm returns binary values so that a buyer would either win a copy of the good, or fail to do so.
Here we consider a more expressive class of mechanisms in which a buyer $i$ may be offered a
probability $\lambda_i$ of receiving the item; assuming that buyers are risk-neutral,
if $i$ has valuation $v_i$ for the item,
then $i$ would have valuation $v_i\lambda_i$ for the probability $\lambda_i$ to receive it.
The general question we consider is, to what extent does the expressiveness of lotteries help us to design truthful mechanisms that better approximate \opt?

\subsection{Model and preliminaries}\label{sec:preli}
We consider a setting in which we want to auction lotteries for digital goods \cite{GHKSW06}, i.e., goods with unlimited supply and no production costs.
A lottery $L = (\lambda, p)$ (for a specified item) is defined by its win probability
$\lambda\in[0,1]$ and its non-negative real-valued price $p$.
A bidder with valuation $v$ purchasing lottery $L$ has \emph{utility} $\lambda v - p$, i.e., his valuation for the item ``weighted'' by the win probability minus the lottery's price.

The model is defined as follows. There are $n$ selfish bidders, with \emph{true} valuations $v_1,\ldots,v_n$, who bid $b_1,\ldots,b_n$ respectively. A mechanism (also called auction) is a pair $(A, P)$: $A$ is an algorithm which on input $\b=(b_1,\ldots,b_n)$ returns a vector of $n$ win probabilities $(\lambda_1(\b), \ldots, \lambda_n(\b))$; $P(\b)$ is a vector of $n$ payment functions, i.e., how much each bidder is charged to buy the lottery. We focus on auctions which \emph{deterministically} propose lotteries $L_i=(\lambda_i(\b),p_i(\b))$ to each bidder $i$; bidders buy the lottery if their utility is non-negative, i.e., the lottery satisfies \emph{voluntary participation}.
Given that win probabilities and payments are deterministic, $i$'s utility $\lambda_iv_i-p_i$ is also deterministic. 
(That is, if the auction was repeated with the same bids, each bidder would be offered the same lottery, and the auctioneer's revenue would be the same. For a buyer, there remains the uncertainty regarding the win/lose outcome of the lottery itself.)

As mentioned in the discussion below about related works, theoretical motivations for studying this model of lotteries are related to the ones in \cite{BCKW10} and \cite{GHKSW06}. From a more practical point of view, the lotteries studied in the paper are perhaps better thought of as a model of uncertainty about what a buyer will receive. For example, in the context of TV advertising, advertisers pay to buy an ad slot for which the size of the exposed audience is uncertain but may be modeled with a probability distribution (via e.g. data on programs' audience shares). Also of relevance is Swoopo-style ``recreational shopping''. We believe that this model may suggest new kinds of sponsored search products, akin to TV advertising ones, that differ from those currently used by search engines.

We aim to design \emph{truthful} auctions (where bidders maximize their utility when telling their true valuations).
In this setting, from the definition of utility, it is immediate to see that we are dealing with 
one-parameter 
bidders 
\cite{AT01}.
Therefore, according to the characterization of truthful mechanisms for one-parameter bidders, we have to design \emph{monotone} lotteries,
i.e., lotteries for which the win probability is non-decreasing in the bid. Moreover, $p_i(\b)$, the amount bidder $i$ must be charged, 
must be of the form~\cite{AT01,GolHar05}:
\[
p_i(b_i, \bi) = h_i(\bi) + b_i \lambda_i(b_i, \bi) - \int_0^{b_i}\lambda_i(b,\bi)db,
\]
where 
$\bi$ denotes the bid vector $\b$ with all but the $i$-th entry and $h_i$ is any function depending on $\bi$ (but not $b_i$).
A stronger requirement is to demand \emph{collusion-resistant} auctions in which \emph{any} coalition of bidders maximizes the sum of the utilities of its members when they are truthtelling. Collusion-resistant lotteries are characterized in this context in terms of \emph{singular} allocations. A monotone win probability function $\lambda_i$ is singular if for all $i, b_i, \bi', \bi$, we have $\lambda_i(b_i, \bi)=\lambda_i(b_i, \bi')$.\footnote{Goldberg and Hartline
\cite{GolHar05} call these allocations \emph{posted-price mechanisms}. However, we prefer the name singular as it reflects better the semantics of the property in the context of lotteries.} A singular win probability $\lambda_i$ is then a function of $b_i$ only. Goldberg and Hartline
\cite{GolHar05} prove
the following theorem.

\begin{proposition}[\cite{GolHar05}]\label{prop:singular}
A lottery is collusion-resistant if and only if its win probability functions are singular.
\end{proposition}

We want incentive-compatible auctions which guarantee a good approximation of the optimal revenue of the auctioneer.
The optimal revenue is defined as $\opt = \sum_i v_i$ for bidders' valuations $v_1, \ldots, v_n$.
(For a given bid vector $\b$, we let $\opt(\b)=\sum_{i=1}^n b_i$.)
A mechanism approximates \opt\ within a ratio $r$ if the sum of the payments collected from the bidders is at least $\opt/r$.
Alternative benchmarks for comparisons are used in the literature, the most prominent being ${\cal F}^{(2)} = \max_{i=2,\ldots,n} i v_i$, assuming $v_i \geq v_{i+1}$~\cite{GHKSW06}.
${\cal F}^{(2)}$ measures the maximum revenue achievable with a fixed posted price, under the constraint that at least two items are sold. For ${\cal F}(\b)$ defined as $\max_{i=1,\ldots,n} i b_i$, with $b_i \geq b_{i+1}$, it is known that ${\cal F}(\b) = \Theta(\opt(\b)/\ln(n))$ and that for all $\b$, ${\cal F}(\b) \geq \opt(\b)/\ln(n)$~\cite{GHKSW06}. The focus on the rather technical benchmark $\ftwo$ is mainly motivated by the impossibility to approximate $\opt$ reasonably well (cf.~\cite{GHKSW06}).

The following observation indicates the 
power given by lotteries in revenue maximization. If bidders are known to have 2 possible types, represented by a high valuation $H$ and a low valuation $L$, then \opt\ cannot be approximated better than about $H/L$ by any deterministic ``classical'' auction selling goods and not lotteries. By contrast, if the auctioneer is allowed to sell lottery tickets with win probabilities of $1/2$ and $1$ respectively, then a bidder with valuation $L$ would be willing to pay $L/2$ to buy the first kind of ticket; similarly, a bidder with valuation $H$ will pay $H/2$ to buy the ticket with win probability $1$. The revenue of such a lottery is then $\opt/2$.

\subsection{Our contribution}
Our results are summarized in Table \ref{table}.
When bids are known to come from a finite domain of size $d$, we establish that the optimal revenue may be approximated within factor $d$, but no lower constant factor is possible. Moreover, the upper bound of $d$ is attained by a straightforward mechanism that is both \emph{anonymous} (that is, an offer depends only on an agent's bid, and not on his identity) and singular (and so collusion-resistant). Meanwhile the lower bound applies to all truthful mechanisms, regardless of computational considerations. It is important to note that the lower bound identifies, for any $d$ and $\epsilon>0$, a set of ``well separated'' $d$ bids for which no truthful lottery can approximate $\opt$ better than $d-\epsilon$.\footnote{The fact that the bids identified by Theorem \ref{thm:finite:LB} are quite far apart from each other explains why for $D=[1,H]$ it is possible to approximate $\opt$ within $\ln(H)+1$.}

\begin{table*}[!t]
\centering
\begin{tabular}{l|c|c|}
\cline{2-3} & \rule{0ex}{14pt} Upper bound & Lower bound \\
\hline
\multicolumn{1}{|c|}{$D=\{L,H\}$} & \rule{0ex}{14pt} $\frac{2H-L}{H}^\star$ (Thm \ref{thm:2value:UB}) & $\frac{2H-L}{H}^\#$ (Thm \ref{thm:2value:LB}) \\
\hline
\multicolumn{1}{|c|}{$D=\{B_1, \ldots, B_d\}$} & \rule{0ex}{14pt} $d^\star$ (Thm \ref{thm:finite:UB}) & $d-\epsilon^\#$, any $\epsilon>0$ (Thm \ref{thm:finite:LB}) \\
\hline
\multicolumn{1}{|c|}{$D=[1,H]$} & \rule{0ex}{14pt} $\ln H+1^\star$ (Thm \ref{thm:cont:UB}) & $\ln(H)+1^\#$ (Thm \ref{thm:cont:lowerbound})\\
\hline
\end{tabular}
\caption{The bounds on the approximation guarantee of the revenue of incentive-compatible lotteries versus \opt\ as a function of bidders' domain; the bounds marked by `$\star$' hold for collusion-resistant lotteries, those marked by `$\#$' apply to truthful lotteries.\label{table}}
\end{table*}

The motivation to study two-value domains $\{L,H\}$, $L<H$, arises from the fact that one can provably achieve the best approximation guarantee with respect to $\ftwo$ under this assumption 
\cite{RSOPec09} (see below for more details). 
More generally, many real-life applications involve bidders with valuations from a finite domain. Money is, by its very nature, discrete with reasonable lower and upper bounds. Similarly, auctions on the web may collect bids through drop-down menus; the values available define a finite domain.

Regarding bids that may come from a {\em continuous} domain $[1,H]$, we obtain a tight bound of $\ln(H)+1$, where again the upper bound is obtained via a simple collusion-resistant lottery and the lower bound holds for any truthful lottery.

Our lower bounds measure the limitation of truthful mechanisms in terms of approximation guarantee to \opt\ independently of the number of bidders. Surprisingly, the best one can achieve when requiring incentive-compatibility is in fact obtainable by mechanisms having the stronger property of collusion-resistance. To the best of our knowledge, this is the first known case in which collusion-resistant mechanisms are as strong as truthful ones. This represents, in a sense, a first positive result on (deterministic) collusion-resistant mechanisms which have otherwise very limited power, as shown by the strong negative results in \cite{GolHar05,Sch00}. Let us note that, regardless of the size of the domain, there always exists a bid vector for which $\ftwo=\opt$ (e.g., a vector of all $L$ for binary domains) and therefore the approximation guarantees of our lotteries w.r.t. $\ftwo$ do not improve over known results in the worst case. (Nevertheless, this is not surprising as our lotteries are resistant to collusive behavior and not just truthful as the auctions designed to approximate $\ftwo$.) We leave open the problem of devising incentive-compatible auctions tailored to approximate $\ftwo$ since our main aim here is different: firstly, using \opt\ as benchmark and, secondly, showing the equivalence between two notions of incentive compatibility at different ends of the spectrum.

We use two different techniques to prove our results. The lower bound for continuous domains is obtained by adopting a probabilistic technique designed in \cite{GHKSW06} to bound from below the revenue of \emph{universally truthful} auctions. These are auctions defined as probability distributions over deterministic truthful auctions. To use the probabilistic technique of \cite{GHKSW06} we prove that there is a bijection between truthful lotteries and universally truthful auctions. However, the technique of \cite{GHKSW06} turns out to be not flexible enough to prove bounds which do not depend on the actual values in $d$-sized domains. Therefore,  Theorem~\ref{thm:finite:LB} needs the development of a new technique that we introduce in the simpler setting of two-value domains in the proof of Theorem~\ref{thm:2value:LB}. (The latter result can also be proved via the probabilistic proof technique given in \cite{GHKSW06}.) 
Our new lower bounding technique relies on the application of Carver's theorem \cite{Craver22} which characterizes inconsistent linear inequality systems in terms of certain linear combinations of the constraints of the system. Requiring to approximate \opt\ within a given ratio gives rise to a linear system with a particular structure; our proofs exploit this structure to define a suitable linear combination according to Carver's theorem. In setting up the system of linear inequalities we are able to ``hide'' the details given by the values in the domain and we can only focus on the asymptotic behavior of lotteries. Technically less involved proofs of Theorems~\ref{thm:2value:LB} and \ref{thm:finite:LB} would suffice to obtain the corresponding results in the case of fixed number of bidders 
but would weaken the equivalence of truthful and collusion-resistant mechanisms to hold only in rather limited scenarios. 

Finally, let us note that due to the simplicity of our anonymous and singular lotteries, our upper bounds hold also in the online setting of \cite{KouPie10}, i.e., they hold in a setting 
in which bidders come online and a decision on the lottery to offer has to be made before the next bidder arrives.

\paragraph{Roadmap.} In Section \ref{sec:related} we review the related literature and, in particular, compare our model with that considered in previous revenue-maximizing auctions for digital goods. In Section \ref{sec:bounded:UB} we give the straightforward logarithmic upper-bound on approximability of \opt\ through lotteries when bidders bid from an interval $[1,H]$. Domains comprised of only two values are considered in Section \ref{sec:2value}. The results in this section are extended to any finite domain in Section \ref{sec:finite}. Finally, we consider the relation between our notion of lotteries and the concept of universally truthful auctions in Section \ref{sec:equivalence}, and use this result to prove the matching lower bound for continuous domains $[1,H]$. Section \ref{sec:conclusions} contains some concluding remarks.

\subsection{Related works}\label{sec:related}
This work is motivated by the results in \cite{BCKW10}. Briest et al. \cite{BCKW10} show how lotteries help in maximizing the revenue when designing envy-free prices. Here we address a similar type of question and aim at obtaining similar results for incentive-compatible lotteries.

Truthful lotteries defined above naturally relate to the truthful auctions for digital goods considered in \cite{GHKSW06}. The authors of
\cite{GHKSW06} show that no deterministic truthful auction can guarantee a reasonable approximation of ${\cal F}$ and therefore focus on {universally truthful} auctions. However, they also show that these auctions fail to guarantee any constant approximation of $\cal F$ (cf. Lemma 3.5 in \cite{GHKSW06}) and therefore the benchmark of interest becomes $\ftwo$. They define an interesting auction called Random Sampling Optimal Price (RSOP, for short) and prove that RSOP gives a (quite weak) constant approximation of ${\cal F}^{(2)}$; they also conjecture the right constant to be $4$. Better bounds are then proved in \cite{FFHK10,RSOPec09}; the latter work proves the conjecture 
when the number of winners is at least 6 and in general for two-valued domains.

Our lotteries are deterministic, although a degree of randomness is given by the probabilistic nature of the allocation. This random aspect can be seen to imply that our truthful lotteries are in fact equivalent to universally truthful auctions (e.g., the lottery of the example above can be seen as an uniform probability distribution over a deterministic auction which charges $L$ and one which charges $H$ for the item). Nevertheless, our equivalence proof also shows how lotteries can be seen as a (arguably more intuitive) reinterpretation of universally truthful auctions (e.g., the universally truthful counterpart of the lottery for $[1,H]$ domain involves a probability distribution over an infinite number of mechanisms). (See discussion at the end of Section~\ref{sec:equivalence} for details.) Because of the aforementioned equivalence, this research can also be seen as a continuation of \cite{GHKSW06} studying to which extent the knowledge of the domain helps in approximating \opt\ (e.g., Theorems \ref{thm:2value:UB} and \ref{thm:finite:UB} contradict the inapproximability result in \cite{GHKSW06}, i.e., Lemma 3.5 therein breaks down for finite domains).  

Another related work is \cite{GolHar05} which considers collusion-resist\-ant mechanisms for bidders with domains similar to ours. A characterization in terms of allocation rules is given (cf. Proposition \ref{prop:singular})
for randomized mechanisms which define allocations, payments and then utilities in expectation over the random coin tosses of the mechanism. Since in our lotteries we work with deterministic utilities (even, when considering lotteries as universally truthful auctions) this characterization 
holds also in our setting. 

A related concept is the {\em responsive lotteries}, studied in~\cite{FT10}, in which a single agent reports his valuations of a set of alternatives, and is awarded one of them, using probabilities designed to incentivize him to report his true valuations (up to affine rescaling). The difference here is that we have just one kind of item, and multiple agents.

Hart and Nisan \cite{HN12} use a model similar to ours (risk-neutral bidders and lottery offers) in their study of the optimal revenue when selling multiple items.

Other benchmarks are defined in the literature to compare the revenue of incentive-compatible auctions, see, e.g., \cite{Ngu11}. To the best of our knowledge, our work is the first in which revenue is compared to \opt.\footnote{However, in certain combinatorial settings, such as the one considered in e.g.~\cite{BBM08}, \opt, as social welfare, is used as benchmark for revenue maximization.}

\section{A simple collusion-resistant lottery}
\label{sec:bounded:UB}
In this section we assume that the bids belong to the interval $[1,H]$. We begin by proving an upper-bound on the revenue guarantee by a collusion-resistant lottery. We define $\lambda_i(\b)=\frac{\ln(e \cdot b_i)}{\ln(e \cdot H)},$  for any $i$ and $\b$. Therefore, since
$$
 \int_0^{b_i} \ln (e \cdot u) du = b_i \ln (b_i)
$$
we have
$$
 p_i(\b) = h_i(\bi) + \frac{1}{\ln(e \cdot H)}\left(\rule{0ex}{12pt} b_i \ln (e \cdot b_i) - b_i \ln (b_i)\right) = h_i(\bi) + \frac{b_i}{ \ln(e \cdot H)}.
$$
Setting $h_i(\bi)=0$ (the value of the functions $h_i$ has no consequence for the truthfulness of the auction given that it is independent of $b_i$), we have that the utility of agent $i$, when declaring $b_i$, is
$$
v_i \frac{\ln(b_i)}{\ln(e \cdot H)} - \frac{b_i}{\ln(e \cdot H)}
$$
which is maximized for $b_i=v_i$. Since each bidder pays a fraction $1/\ln(e \cdot H)$ of his bid, the revenue of this truthful auction is a $(\ln(H)+1)$-approximation of \opt.  Finally, observe that since $\lambda_i(\b)=\lambda(b_i)$ then the lottery is anonymous and singular and therefore collusion-resistant. Thus, we have the following result.

\begin{theorem}\label{thm:cont:UB}
There exists a $(\ln(H)+1)$-approximate anonymous collusion-resistant lottery for bidders whose valuations belong to the interval $[1,H]$.
\end{theorem}

We defer to Section \ref{sec:equivalence}, the proof of a matching lower bound for any truthful lottery and bidders having $[1,H]$ domains.

\section{Two-value domains}\label{sec:2value}
In this section we assume that bidders' valuations are known to come from a 2-element set $\{L,H\}$, with $L<H$.
To begin with, we understand how to define lambda functions that lead to truthful lotteries and maximize the revenue. It turns out that some simple lambda functions that correspond to anonymous, collusion-resistant mechanisms guarantee a good approximation of the optimal revenue. We then prove (Theorem~\ref{thm:2value:LB}) that functions of this kind suffice to obtain optimal performance, even amongst mechanisms that need not be collusion-resistant, but just truthful.

We initially observe that by essentially the same arguments in~\cite{AT01} it is easy to show that a necessary condition to obtain truthfulness
is that $\lambda_i(H, \bi) \geq \lambda_i(L, \bi)$ for any $i$ and $\bi$.\footnote{We cannot blindly use the results in \cite{AT01} since the technical
assumption in that work is to have the lambda functions twice differentiable. This is not true for discontinuous functions,
like the lambda functions for two-value domains.} For a bidder $i$ with valuation $H$, the following truthfulness
constraint must be satisfied
$
H \lambda_i(H, \bi) - p_i(H, \bi) \geq H \lambda_i(L, \bi) - p_i(L, \bi),
$
which implies that $p_i(H, \bi) \leq H(\lambda_i(H, \bi)-\lambda_i(L, \bi)) + p_i(L, \bi)$.
To maximize the revenue we would like to set $p_i(H, \bi)=H(\lambda_i(H, \bi)-\lambda_i(L, \bi)) + p_i(L, \bi)$.
We next show that we can do that while guaranteeing the truthfulness of bidder $i$ having valuation $L$.
Indeed, the utility of such a bidder when lying and declaring $H$ is:
\begin{align*}
-p_i(H, \bi)+L \lambda_i(H)=-H(\lambda_i(H)-\lambda_i(L)) + L \lambda_i(H)-p_i(L, \bi) \leq L\lambda_i(L) - p_i(L, \bi)
\end{align*}
where the last inequality follows from $(L-H)(\lambda_i(H, \bi)-\lambda_i(L, \bi))\leq 0$. It remains to set a value for $p_i(L, \bi)$ to guarantee voluntary participation, i.e., to guarantee that a bidder with valuation $L$ buys the lottery; to achieve this, we set $p_i(L)=L\lambda_i(L, \bi)$.

From the above analysis, one could easily get an approximation guarantee of $2$ by setting $\lambda_i(H, \bi)=1$ and $\lambda_i(L, \bi)=1/2$ for any $i$ and $\bi$.
However, it is possible to do better as shown by the next theorem.
Below, we let $N_H(\b)$ (resp. $N_L(\b)$) be the set of bidders declaring $H$ (resp. $L$) in $\b$ and $n_H(\b)=|N_H(\b)|$;
similarly, $n_L(\b)$ denotes the number of $L$'s in $\b$.

\begin{theorem}\label{thm:2value:UB}
There exists an anonymous collusion-resistant lottery for two-value domains $\{L,H\}$, $L<H$, whose revenue is a $\frac{2H-L}{H}$-approximation of \opt.
\end{theorem}
\begin{proof}
We define $\lambda_i(L, \bi)=\frac{H}{2H-L}$ and $\lambda_i(H, \bi)=1$ for all $i$ and $\bi$. With the payment functions defined above, which guarantee truthfulness, the revenue collected by this lottery for a vector $\b$ is:
\begin{align*}
\sum_{i=1}^n p_i(\b) &= \sum_{i \in N_H(\b)} p_i(b_i, \bi) + \sum_{i \in N_L(\b)} p_i(b_i, \bi) \\ &= \sum_{i \in N_H(\b)} \left(H(\lambda_i(H, \bi)-\lambda_i(L, \bi)) + L \lambda_i(L, \bi)\rule{0ex}{3ex}\right) + \sum_{i \in N_L(\b)} L \lambda_i(L, \bi)\\
&= n_H(\b) H \left(1 - \frac{H}{2H-L}\right) + n_H(\b) L \frac{H}{2H-L} + n_L(\b) L \frac{H}{2H-L} \\
& = \left(n_H(\b)H+n_L(\b)L\rule{0ex}{3ex}\right) \frac{H}{2H-L}.
\end{align*}
The approximation guarantee follows from the observation that the optimum is defined as $n_H(\b)H+n_L(\b)L$. Finally, the collusion-resistance and the anonymity of the lottery follows from having defined $\lambda_i(\b)=\lambda(b_i)$ for all $i$ and $\b$.
\end{proof}

Next we show that no truthful lottery can improve on the ratio obtained by Theorem~\ref{thm:2value:UB}.
Lotteries that are truthful --- but not necessarily collusion-resistant --- allow $\lambda_i$ to depend also on the bid vector $\bi$.
The following result can be proved with the same logic in the case of $n=1$ (where the vector $\bi$ is trivial) with a significantly shorter proof\footnote{In such a case, the graph in Figure~\ref{fig:LB} would comprise only two nodes and then the combinatorics involved is straightforward. (Moreover, $n=1$ is of little significance given that there is no difference between truthfulness and collusion-resistance in this case.)}; we show here that it holds for any number of bidders $n$.  The proof follows a simple structure. We initially upper bound the payments of a truthful lottery in terms of the lambda functions. Then, for the lottery to approximate $\opt$ within a factor $\alpha$, a certain system of linear inequalities, with the lambda functions as variables, must be satisfied. Finally, we study the combinatorics of the system and use Carver's theorem \cite{Craver22} to determine the values of $\alpha$ for which the system admits solutions.

\begin{theorem}\label{thm:2value:LB}
No truthful lottery for two-value domains $\{L,H\}$, $L<H$, has approximation guarantee better than $\frac{2H-L}{H}$.
\end{theorem}
\begin{proof}
Consider a truthful lottery which has approximation guarantee better than $\alpha\geq 1$ for a 2-value domain $\{L,H\}$, $L<H$.
By voluntary participation for bidders with valuation $L$ we have
\begin{equation}\label{eq:vpl}
p_i(L, \bi)  \leq L \lambda_i(L, \bi).
\end{equation}
By truthfulness of bidders with valuation $H$, we have $p_i(H, \bi) - p_i(L, \bi) \leq H(\lambda_i(H,\bi)-\lambda_i(L,\bi))$, equivalently
\begin{equation}\label{eq:trh}
p_i(H, \bi)  \leq H\bigl(\lambda_i(H,\bi)-\lambda_i(L,\bi)\bigr)+p_i(L, \bi).
\end{equation}
Combining (\ref{eq:vpl}) and (\ref{eq:trh}) we get
\begin{equation}\label{eq:2value:paymentsUB}
p_i(H, \bi) \leq H \lambda_i(H,\bi)-(H-L)\lambda_i(L,\bi).
\end{equation}
From (\ref{eq:2value:paymentsUB}) we have that for all bid vectors $\b$,
$$
\sum_{i \in N_H(\b)} \bigl(H \lambda_i(H,\bi)- (H-L) \lambda_i(L,\bi)\bigr)+\sum_{i \in N_L(\b)} L \lambda_i(L, \bi) \geq \sum_{i=1}^n p_i(\b).
$$
Now, noting that $\sum_{i=1}^n p_i(\b) > \frac{\opt(\b)}{\alpha} = \frac{n_H(\b) H + n_L(\b)L}{\alpha}$,
we can rewrite this as
\[
\sum_{i \in N_H(\b)} \left(H \lambda_i(H,\bi)- (H-L) \lambda_i(L,\bi)\right)+\sum_{i \in N_L(\b)} L \lambda_i(L, \bi) > \frac{n_H(\b) H + n_L(\b)L}{\alpha}.
\]
%
Rearranging the above and noting that $\lambda_i(H,\bi) \leq 1$ for all $i$ and $\bi$, we have the following system of $2^n$ linear inequalities:
\begin{align}\label{eq:2value:LB:ineqsys}
-(H-L) \sum_{i \in N_H(\b)} \lambda_i(L,\bi) + L \sum_{i \in N_L(\b)} \lambda_i(L, \bi) > -\frac{n_H(\b)\alpha - n_H(\b)}{\alpha} H + \frac{n_L(\b)}{\alpha} L \ \ \ \textrm{ for all } \b.
\end{align} 
There are $n \cdot 2^{n-1}$ variables $\{\lambda_i(L,\bi):i\in[n],\bi\in\{L,H\}^{n-1}\}$ in the above system.
Notice that each variable $\lambda_i(L, \bi)$ occurs only twice with a coefficient different from $0$: in particular, its coefficient is $0$ in all the constraints relative to a bid vector $\b=(\cdot, \bi')$ with $\bi' \neq \bi$; for $\b=(H, \bi)$, $\lambda_i(L, \bi)$ has coefficient $-(H-L)$ since $i \in N_H(\b)$; finally, for $\b=(L, \bi)$, the variable has coefficient $L$ since $i \in  N_L(\b)$.

In order to prove the theorem we want to study the values of $\alpha$ for which this system has no solutions. Towards this end, we let $x_i(\bi)$ be a shorthand for $\lambda_i(L,\bi)$ and number all the possible $m=2^n$ bid vectors $\b$ as $\b^{(1)}, \ldots, \b^{(m)}$. Then we denote by $\Gamma^{(j)}(x)$ the terms involved in the $j$-th constraint of (\ref{eq:2value:LB:ineqsys}) relative to $\b^{(j)}$, i.e.,
\[
\Gamma^{(j)}(x) := -(H-L) \sum_{i \in N_H(\b^{(j)})} x_i(\bi^{(j)}) + L \sum_{i \in N_L(\b^{(j)})} x_i(\bi^{(j)}) + \frac{n_H(\b^{(j)})\alpha-n_H(\b^{(j)})}{\alpha}H - \frac{n_L(\b^{(j)})}{\alpha}L.
\]
We use Carver's theorem \cite{Craver22} which characterizes inconsistent systems of linear inequalities: according to \cite{Craver22}, (\ref{eq:2value:LB:ineqsys}) above admits no solution if and only if we can find $m+1$ non-negative constants $k_j$, such that
\begin{equation}\label{eq:2value:LB:carverineq}
\sum_{j=1}^{m} k_j \Gamma^{(j)}(x) + k_{m+1} \equiv 0,
\end{equation}
with at least one of the $k$'s being positive. We next show how to define the constants $k_1,\ldots, k_m$ so that the two occurrences with non-null coefficients of each variable $x_i(\bi)$ cancel out. 

\begin{figure}[t]
\centering
\includegraphics[scale=0.7]{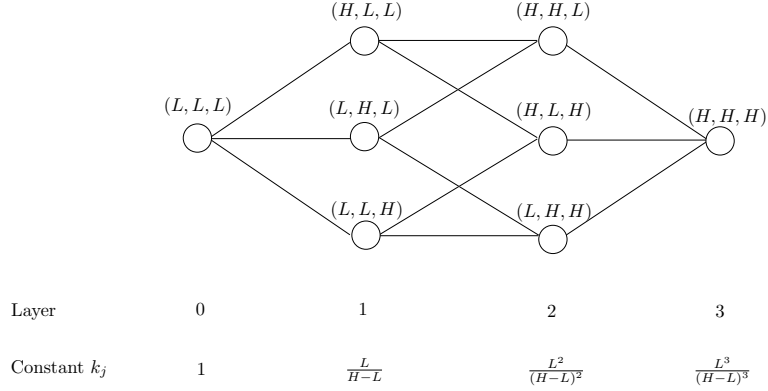}
\caption{The graph constructed for $n=3$.}\label{fig:LB}
\end{figure}

To show how to define the constants $k_j$ we define a graph which has a vertex for each possible bid vector. We put an edge between two vertices if the corresponding bid vectors are \emph{adjacent}, i.e., they differ in only one entry. Each vertex then has exactly $n$ neighbors. This is a layered graph with $n+1$ layers. Let layer $\ell$ be the set of all vertices whose corresponding bid vectors are comprised of $\ell$ $H$'s and $n-\ell$ $L$'s. The graph is indeed layered as by definition a node at layer $\ell$ only has neighbors at layer $\ell-1$ and $\ell+1$. For a bid vector $\b^{(j)}$ whose associated vertex lies in layer $\ell$ of the graph we define $k_j=\frac{L^\ell}{(H-L)^{\ell}}$. The construction is depicted in Figure \ref{fig:LB} for the case of $n=3$. To show that this definition of constants cancels out the $x$'s in (\ref{eq:2value:LB:carverineq}) consider a variable $x_i(\bi)$. The two occurrences of $x_i(\bi)$ with a non-zero coefficient are for the bid vectors $\b_1=(H,\bi)$, in which case $x_i(\bi)$ has coefficient $-(H-L)$, and $\b_2=(L,\bi)$, in which case $x_i(\bi)$ has coefficient $L$. The two vertices associated to $\b_1$ and $\b_2$ are by definition connected, with $\b_1$ being a node at layer $n_H(\b_1)$ and $\b_2$ being a node at layer $n_H(\b_2)$. By construction, we have $n_H(\b_2)=n_H(\b_1)-1$ and therefore, letting $\ell=n_H(\b_1)-1$, we have that the constant $k$ corresponding to $\b_1$ is $\frac{L^{\ell+1}}{(H-L)^{\ell+1}}$ while the constant corresponding to $\b_2$ is $\frac{L^{\ell}}{(H-L)^{\ell}}$. This means that the contribution of variable $x_i(\bi)$ to (\ref{eq:2value:LB:carverineq}) is: $$\left(L \frac{L^\ell}{(H-L)^{\ell}} - (H-L) \frac{L^{\ell+1}}{(H-L)^{\ell+1}}\right)x_i(\bi)=0.$$ Then, the following holds:
\[
\sum_{j=1}^m k_j \Gamma^{(j)}(x) = \sum_{j=1}^m k_j \Gamma^{(j)} = \sum_{j=1}^m k_j \left(\frac{n_H(\b^{(j)})\alpha -  n_H(\b^{(j)})}{\alpha} H - \frac{n_L(\b^{(j)})}{\alpha} L\right).
\]
We now study for which values of $\alpha$ the above sum is not positive. 

Consider all the nodes at layer $\ell$ of the graph. We call this set of nodes $S_\ell$. We abuse the notation and say that $\b \in S_\ell$ to mean that the node corresponding to $\b$ is at layer $\ell$ of the graph; we then rewrite the summation above and impose it to be less than or equal than $0$:
\begin{align*}
\sum_{j=1}^m k_j \Gamma^{(j)} = \sum_{\ell=2}^{n+1} \left( -\sum_{\b^{(j)} \in S_{\ell-1}} k_j \frac{n_L(\b^{(j)})}{\alpha}L + \sum_{\b^{(j)} \in S_{\ell}} k_j \frac{\alpha n_H(\b^{(j)}) - n_H(\b^{(j)})}{\alpha} H \right) \leq 0.
\end{align*}
If the whole summation is non-positive, then there exists at least one layer $\ell$ for which the inner summation is non-positive. The bid vectors $\b$ in $S_\ell$ have exactly $\ell$ $H$'s, i.e., $n_H(\b)=\ell$ and $n_L(\b)=n-\ell$, and then their number is $n \choose \ell$. Moreover,  all those bid vectors have the constant $k$ set to $(L/(H-L))^\ell$. Therefore, we have
{\allowdisplaybreaks
\begin{align*}
0 & \geq -\sum_{\b^{(j)} \in S_{\ell-1}} k_j \frac{n_L(\b^{(j)})}{\alpha}L + \sum_{\b^{(j)} \in S_{\ell}} k_j \frac{\alpha n_H(\b^{(j)}) - n_H(\b^{(j)})}{\alpha} H &
\\ & = - {n \choose {\ell-1}} \frac{L^{\ell-1}}{(H-L)^{\ell-1}} \frac{n-\ell+1}{\alpha}L +{n \choose {\ell}} \frac{L^{\ell}}{(H-L)^{\ell}} \frac{\alpha \ell - \ell}{\alpha} H & \Longleftrightarrow \\0 & \geq  -\frac{n(n-1)\cdots(n-\ell+1)}{(\ell-1)!} \frac{L}{\alpha} + \frac{n(n-1)\cdots(n-\ell+1)}{\ell!} \frac{L}{(H-L)} \frac{\alpha \ell -\ell}{\alpha} H &\Longleftrightarrow \\
0 & \geq - \frac{1}{\alpha} + \frac{1}{(H-L)} \frac{\alpha -1}{\alpha} H,
\end{align*}}%
which implies that $\alpha \leq \frac{2H-L}{H}$. This means that for these values of $\alpha$, the weighted sum of the known terms in (\ref{eq:2value:LB:carverineq}) is non-positive. Therefore, there exists a non-negative constant $k_{m+1}$ which together with the constants $k_1, \ldots, k_m$ defined above satisfies (\ref{eq:2value:LB:carverineq}). In other words, for $\alpha \leq \frac{2H-L}{H}$, the system (\ref{eq:2value:LB:ineqsys}) has no solutions and therefore no truthful lottery can guarantee better approximation ratios.
\end{proof}

\section{The case of finite domains}\label{sec:finite}
Following the same approach used for two-value domains, one could study three-valued domains $\{L,M,H\}$, $L<M<H$. With such a study, one would prove an upper bound of $\frac{3HM-HL-M^2}{HM}$ (this is done by setting $\lambda_i(H)=1, \lambda_i(M)=\frac{2MH+LH}{3HM-HL-M^2}$ and $\lambda_i(L)=\frac{HM}{3HM-HL-M^2}$) and a matching lower bound (for, e.g., $n=2$ this is not hard to check using Carver's theorem). However, we prefer to focus on asymptotic bounds (on the approximability of \opt) in terms of the number of allowed bid values in the domain, as opposed to detailed bounds in terms of those values. That is the goal of this section.
We begin by proving that it is possible to design collusion-resistant lotteries collecting a $1/|D|$ fraction of the optimal revenue when bidders bid from a finite domain $D$.

\begin{theorem}\label{thm:finite:UB}
There exists an anonymous collusion-resistant lottery for finite domains $\{B_1, \ldots, B_d\}$, $B_1> \ldots > B_d$, whose revenue is a $d$-approximation of \opt.
\end{theorem}
\begin{proof}
We define the anonymous, singular lottery $\lambda(B_i)=(d-i+1)/d$ and the corresponding payment functions $p_i(B_t, \bi)=\sum_{j=t}^d B_j(\lambda(B_j)-\lambda(B_{j+1}))=\sum_{j=t}^d \frac{B_j}{d}$ for $1 \leq t < d$ where we set $\lambda(B_{d+1})=0$. Let us show that these payment functions indeed lead to collusion-resistance. We have to prove that for any bidder $i$ with valuation $B_t$ and any $B_h \in D$:
\[
B_t \lambda(B_t) - p_i(B_t, \bi) \geq  B_t \lambda(B_h) - p_i(B_h, \bi).
\]
By definition this is equivalent to proving:
\[
B_t (h-t) - \sum_{j=t}^d B_j + \sum_{j=h}^d B_j \geq 0.
\]
Now, whenever $h \leq t$, the left-hand side of the above inequality is equal to $B_t(h-t) - \sum_{j=t}^{h-1} B_j$ which is non-negative since $B_j \leq B_t$, $j \geq t$. Similarly, in the case in which $h > t$, the left-hand side of the above inequality equals $B_t(h-t) - \sum_{j=h}^{t-1} B_j$ which is non-negative since $B_j \geq B_t$, $j \leq t$.

To prove the approximation guarantee we note that for all bid vectors $\b$, we have
\begin{align*}
\sum_{i=1}^n p_i(\b) &= \sum_{t=1}^d \sum_{i \in N_t(\b)} p_i(\b) \geq \sum_{t=1}^d n_t(\b) B_t \left(\lambda(B_t)-\lambda(B_{t+1}\rule{0ex}{3ex})\right) = \sum_{t=1}^d \frac{n_t(\b) B_t}{d} = \frac{\opt(\b)}{d},
\end{align*}
$N_t(\b)$ denoting the set of bidders bidding $B_t$ in $\b$ and $n_t(\b)$ being the size of $N_t(\b)$.
\end{proof}

We cannot improve over the above result even by relaxing the collusion-resistance to truthfulness, as shown by next theorem. The proof of the lower bound below is not a simple extension of the arguments used for two-value domains. Although the structure of the proof is similar, the difficulty here rests on the fact that the layered graph is not adequate to model the definition of the constants required by Carver's theorem. Moreover, the study of the weighted sum of the known terms of the system is significantly more involved. 

\begin{theorem}\label{thm:finite:LB}
For any $d$ and $\epsilon>0$, there exist bids $B_1>B_2> \ldots>B_d>0$ such that no truthful lottery for the domain $D=\{B_1, \ldots, B_d\}$ has approximation guarantee better than $d-\epsilon$.
\end{theorem}
\begin{proof}
By truthfulness constraints any truthful lottery must satisfy the following upper bounds on the payments, for $1 \leq t <d$.
Recall that $p_i$ and $\lambda_i$ map bid vectors to bidder $i$'s payment and win probability, respectively.
\begin{align*}
p_i(B_t, \bi) & \leq B_t (\lambda_i(B_t,\bi)-\lambda_i(B_{t+1},\bi))+p_i(B_{t+1}, \bi) \\ & \leq B_t \lambda_i(B_t, \bi) + \sum_{j=t+1}^d \left( -(B_{j-1} - B_{j}) \lambda_i(B_{j}, \bi) \right)\\
& = \sum_{j=t}^{d-1} B_{j}\left( \lambda_i(B_{j}, \bi) - \lambda_i(B_{j+1}, \bi) \right) + B_d \lambda_i(B_d, \bi)\\
& \leq B_t (\lambda_i(B_t,\bi)- \lambda_i(B_{t+1},\bi))+ \sum_{j=t+1}^{d} B_{j},
\end{align*}
where in the second inequality we recursively use the bound obtained on $p_i(B_t, \bi)$ to $p_i(B_{t+1}, \bi)$, $p_i(B_{t+2}, \bi)$ and so on, and where the last inequality follows from the fact that $0 \leq \lambda_i(B_{j}, \bi) - \lambda_i(B_{j+1}, \bi) \leq 1$ given by truthfulness and by definition, respectively. We let $B_d$ be a positive value and then define the bids of the domain to satisfy\footnote{Note that this is a feasible definition as for the meaningful values of $\epsilon$, i.e., $d > \epsilon>0$, $\frac{\rho_\epsilon}{d-t-1} < 1$ and therefore we are only quantifying the ``gap'' between two consecutive bids of the domain.} $\frac{B_{t}}{B_{t-1}}\leq \frac{\rho_\epsilon}{d-t+1}$, for $1 < t \leq d$, where $\rho_\epsilon := \frac{\epsilon}{d(d-\epsilon)}$ and get that, for $1 \leq t <d$, it holds
\[
p_i(B_t, \bi) \leq B_t (\lambda_i(B_t,\bi)- \lambda_i(B_{t+1},\bi))+  \rho_\epsilon B_{t},
\]
since $\sum_{j=t+1}^d B_j \leq (d-t) B_{t+1} \leq \rho_\epsilon B_t$. We also note that by voluntary participation, we have
$
p_i(B_d, \bi) \leq B_d \lambda_i(B_d,\bi) \leq B_d \lambda_i(B_d,\bi) +  \rho_\epsilon B_{d}.
$
We can now bound from above the revenue of a truthful lottery. To ease the notation we set $\lambda_i(B, \bi)=0$ for any $i, \bi$, and $B \not\in D$. We then get
\begin{align*}
\sum_{i=1}^n p_i(\b) &\leq \sum_{t=1}^d \sum_{i \in N_t(\b)} \left(B_t \left(\rule{0ex}{3ex} \lambda_i(B_t, \bi) - \lambda_i(B_{t+1}, \bi)\right) + \rho_\epsilon B_{t} \right) \\ & = \sum_{t=1}^d \sum_{i \in N_t(\b)}  B_t \left(\rule{0ex}{3ex} \lambda_i(B_t, \bi) - \lambda_i(B_{t+1}, \bi)\right) + \sum_{t=1}^d n_t(\b) B_t \rho_\epsilon,
\end{align*}
where, as above, $N_t(\b)$ is the set of bidders bidding $B_t$ in $\b$ and $n_t(\b)=|N_t(\b)|$. We now assume by contradiction that a truthful lottery has approximation guarantee better than $d-\epsilon$ for the domain $D$ as in the hypothesis. By noticing that $\lambda_i(B_1, \bi) \leq 1$ for all $i$ we then obtain the following system of linear inequalities
\begin{align}\label{eq:finite:LB:ineqsys}
\sum_{t=1}^d \sum_{i \in N_t(\b)}  B_t \left(\rule{0ex}{3ex} \sigma_t \lambda_i(B_t, \bi) - \lambda_i(B_{t+1}, \bi)\right) > -\frac{d-1}{d} n_1(\b) B_1+ \sum_{t=2}^{d} \frac{n_t(\b)B_t}{d}
\end{align}
for all $\b$, where $\sigma_t=1$ if $t>1$ and $0$ otherwise. For $d \geq t>1$, each variable $\lambda_i(B_t,\bi)$ in (\ref{eq:finite:LB:ineqsys}) has only two occurrences with a coefficient different from zero. Indeed, it appears with a coefficient of $-B_{t-1}$ in the constraint relative to the bid vector $(B_{t-1}, \bi)$ and has a factor of $B_t$ in the constraint of $(B_t, \bi)$. The constraints relative to all the other bid vectors have $\lambda_i(B_t,\bi)$ with a zero coefficient. Similarly to the binary-domain case, we enumerate all the possible $m=d^n$ bid vectors, $\b^{(1)}, \ldots,  \b^{(m)}$ and for each of those we define
\begin{align*}
\Gamma^{(j)}(\Lambda) & := \sum_{i=1: i \in N_t(\b^{(j)})}^n B_t \left( \sigma_t \lambda_i(B_t, \bi^{(j)}) - \lambda_i(B_{t+1}, \bi^{(j)})\rule{0ex}{3ex}\right),\\
\Delta^{(j)} & :=\frac{d-1}{d} n_1(\b^{(j)}) B_1 - \sum_{t=2}^d \frac{n_t(\b^{(j)})B_t}{d},
\end{align*}
where $\Lambda=(\lambda_1, \ldots, \lambda_n)$. By Carver's theorem to reach a contradiction and show the theorem it is enough to show that there exist $m+1$ non-negative constants $k_j$, such that
\begin{equation}\label{eq:finite:LB:carverineq}
\sum_{j=1}^{m} k_j (\Gamma^{(j)}(\Lambda)+ \Delta^{(j)}) + k_{m+1} \equiv 0,
\end{equation}
with at least one of the $k$'s being positive. We call these $k$'s the Carver's constants. Akin to the proof of Theorem~\ref{thm:2value:LB}, to prove (\ref{eq:finite:LB:carverineq}) we first show that there exist Carver's constants which make the sum of the $\Gamma$ functions equal to $0$ (cf. Lemma~\ref{le:finite:variables}) and then prove that these constants also annul the sum of the $\Delta$ functions (cf. Lemma~\ref{le:finite:knownterms}).

\begin{lemma}\label{le:finite:variables}
There exist positive constants $k_1, \ldots, k_m$ such that $\sum_{j=1}^{m} k_j \Gamma^{(j)}(\Lambda) \equiv 0$ by canceling out all the variables $\lambda_i(\cdot, \cdot)$. These constants are the same for all the bid vectors $\b^{(j)}$ that have the same value of $\Delta^{(j)}$.
\end{lemma}
\begin{proof}
For any bid vector $\b^{(j)}$, we define
\begin{equation}\label{eq:finite:LB:constants}
k_j=\frac{B_d^{n-n_d(\b^{(j)})}}{\prod_{r=1}^{d-1} B_r^{n_r(\b^{(j)})}}.
\end{equation}

\begin{claim}\label{claim:finite:LB:lambdasprop}
The constants $k_1, \ldots, k_m$ are such that $\sum_{j=1}^{m} k_j \Gamma^{(j)}(\Lambda) \equiv 0$ by canceling out all the variables $\lambda_i(\cdot, \cdot)$ if and only if for any bid vector $\b^{(j)}=(B_t, \bi^{(j)})$, $d \geq t>1$ and $1 \leq i \leq n$, it holds that
\begin{equation}\label{eq:finite:LB:cancelingvars}
k_s = \frac{B_t}{B_{t-1}} k_j,
\end{equation}
where $\b^{(s)}=(B_{t-1}, \bi^{(j)})$.
\end{claim}
\begin{proof}
For the if part, note that by hypothesis $k_1, \ldots, k_m$ in particular cancel out the variable $\lambda_i(B_t, \bi^{(j)})$, $1 \leq i \leq n$ and $d \geq t>1$ (recall that, for all $i$, we bounded from above the variable $\lambda_i(B_1,\cdot)$ with $1$ and then this variable is not in the system). As observed above this variable appears in (\ref{eq:finite:LB:ineqsys}) only twice with a non-null coefficient. Specifically, $\lambda_i(B_t, \bi^{(j)})$ has a coefficient of $B_t$ in the constraint relative to $\b^{(j)}=(B_t, \bi^{(j)})$ and a coefficient of $-B_{t-1}$ in the constraint defined upon $\b^{(s)}=(B_{t-1}, \bi^{(j)})$. Then to cancel $\lambda_i(B_t, \bi^{(j)})$, it must be the case that $B_t k_j = B_{t-1} k_s.$

For the only if part, take Carver's constants $k_1, \ldots, k_m$ which satisfy \eqref{eq:finite:LB:cancelingvars} and assume by contradiction that there exists a variable $\lambda_i(B_t, \bi^{(j)})$, $1 \leq i \leq n$ and $d \geq t>1$, that has not a coefficient $0$ in $\sum_{j=1}^{m} k_j \Gamma^{(j)}(\Lambda)$. Since, as noted above, this variable has a non-zero coefficient only in the constraints relative to $\b^{(j)}=(B_t, \bi^{(j)})$ and $\b^{(s)}=(B_{t-1}, \bi^{(j)})$ respectively then it must be the case that $B_t k_j \neq B_{t-1} k_s,$ thus a contradiction.
\end{proof}
We then need to prove that our definition of Carver's constants in (\ref{eq:finite:LB:constants}) satisfies the requirement in Claim \ref{claim:finite:LB:lambdasprop}. Take two bid vectors $\b^{(j)}$ and $\b^{(s)}$ defined as in the statement of the claim. Since $n_t(\b^{(j)})=n_t(\b^{(s)})+1$, $n_{t-1}(\b^{(j)})=n_{t-1}(\b^{(s)})-1$ and $n_{r}(\b^{(j)})=n_{r}(\b^{(s)})$ for $r \in D \setminus \{t-1, t\}$, by (\ref{eq:finite:LB:constants}), we get that
\begin{align*}
  k_j & = \frac{B_d^{n-n_d(\b^{(j)})}}{\prod_{r=1}^{d-1} B_r^{n_r(\b^{(j)})}} = \frac{B_{t-1}}{B_{t}} \frac{B_d^{n-n_d(\b^{(j)})}}{\prod_{r=1, r\neq t, t-1}^{d-1} B_r^{n_r(\b^{(j)})} \cdot B_t^{n_t(\b^{(j)})-1} \cdot B_{t-1}^{n_{t-1}(\b^{(j)})+1}} \\&= \frac{B_{t-1}}{B_{t}} \frac{B_d^{n-n_d(\b^{(s)})}}{\prod_{r=1}^{d-1} B_r^{n_r(\b^{(s)})}}
  = \frac{B_{t-1}}{B_{t}} k_s.
\end{align*}
The proof concludes by observing that being the bids in the domain positive, so are the Carver's constants we define. Moreover, the second part of the lemma follows from the fact that $\Delta^{(j)}$ is a function of $n_t(\b^{(j)})$, $t \in \{1,\ldots, d\}$. 
\end{proof}

\begin{lemma}\label{le:finite:knownterms}
For the Carver's constants $k_1, \ldots, k_m$ defined in Lemma~\ref{le:finite:variables}, it holds: $$\sum_{j=1}^{m} k_j \Delta^{(j)}= 0.$$
\end{lemma}
\begin{proof}
By definition,
\begin{align*}
\sum_{j=1}^{m} k_j \Delta^{(j)} &= \sum_{\begin{subarray}{c}n_1, \ldots, n_d \in \mathbb{N}:\\ n_1+\ldots+n_d=n\end{subarray}} {n \choose n_1,\ldots, n_d} \frac{B_d^{n-n_d}}{\prod_{r=1}^{d-1} B_r^{n_r}} \left[ \frac{(d-1)}{d} n_1 B_1 - \sum_{h=2}^{d} n_h \frac{B_{h}}{d}\right].
\end{align*}
Observe that for any $1 \leq t \leq d$
\[
{n \choose n_1,\ldots, n_d} n_t = n {n-1 \choose n_1, n_{t-1}, n_t-1, n_{t+1},\ldots, n_d} 
\]
and then
\begin{align*}
\sum_{\begin{subarray}{c}n_1, \ldots, n_d \in \mathbb{N}:\\ n_1+\ldots+n_d=n\end{subarray}} {n \choose n_1,\ldots, n_d} \frac{ n_t B_t}{\prod_{r=1}^{d} B_r^{n_r}} &= n \sum_{\begin{subarray}{c}n_1, \ldots, n_d \in \mathbb{N}:\\ n_1+\ldots+n_t-1+\ldots+n_d=n-1\end{subarray}} {n-1 \choose n_1,\ldots, n_t-1, \ldots, n_d} \frac{1}{B_t^{n_t-1}\prod_{\begin{subarray}{c}r=1,\\r \neq t\,\end{subarray}}^{d} B_r^{n_r}}\\&=n\left(\sum_{i=1}^d \frac{1}{B_i} \right)^{n-1},
\end{align*}
where last equality follows from the multinomial theorem. Consequently, \begin{align*}
\sum_{j=1}^{m} k_j \Delta^{(j)} &= n B_d^n \frac{d-1}{d} \left(\sum_{i=1}^d \frac{1}{B_i} \right)^{n-1} - n B_d^n \frac1d \sum_{t=2}^d \left(\sum_{i=1}^d \frac{1}{B_i} \right)^{n-1}=0.\qedhere
\end{align*}
\end{proof}
By setting $k_{m+1}=0$, the theorem follows by the two lemmata above.
\end{proof}

\section{Universally truthful auctions and lotteries}\label{sec:equivalence}\label{sec:continuous:LB}
Given a function $f(b)$ we call $f^+(b)$ the \emph{right-continuous version} of $f(b)$ defined as
\[
f^+(b) = \lim_{c \rightarrow b^+} f(c).
\] 
where, as usual, $c \rightarrow b^+$ means that $c$ is approaching $b$ from the right. 
\begin{theorem}\label{thm:equiv}
There exists a bijection between truthful lotteries and universally truthful auctions.
\end{theorem}
\begin{proof}
Let $(\lambda, p)$ be a truthful lottery over $[1,H]$. Fix $i$ and $\bi$ and consequently to ease the notation write $\lambda_i(b, \bi)$ as the one-variable function $\lambda(b)$. We show how to define the corresponding universally truthful auction.

We define the function 
\[
\delta(b) = \left\{ \begin{array}{ll} 0 & \textrm{if $b < 1$}\\
											\lambda^+(b)	& \textrm{if $1 \leq b \leq H$}\\
											\lambda(H) & \textrm{if $H<b < 2H$} \\
											1 & \textrm{if $b \geq 2H$},	
											 \end{array} \right.
\]
where $\lambda^+(b)$ is the right-continuous version of $\lambda(b)$. (Note that $\lambda^+(H)=\lambda(H)$ as the domain of $\lambda$ is $[1,H]$.) Observe that $\delta$ is right-continuous, non-decreasing and that $\lim_{b \rightarrow -\infty} \delta(b)=0$, $\lim_{b \rightarrow +\infty} \delta(b)=1$. Therefore, it is known that there exists on some probability space a random variable $X$ for which $\delta(x)=\prob{X \leq x}$ (see, e.g., Theorem 14.1 in \cite{probbook}). In other words, $\delta$ is the cumulative distribution function of a random variable $X$. We can use $X$ to define a universally truthful auction that with 
\begin{equation}\label{eq:mapping}
\prob{x < X \leq y}=\delta(y)-\delta(x) 
\end{equation}
charges a price in $(x,y]$. Moreover, price $z$ is a weak threshold  (i.e., bidder wins by declaring at least $z$) if $\lambda$ is right-continuous at $z$, and strict (i.e., bidder wins by declaring strictly more than $z$) otherwise, cf. Definition 2.4 in \cite{GHKSW06}. 

Conversely, starting from the cumulative distribution function of a universally truthful auction we can define a truthful lottery by using the arguments above backwards.
\end{proof}

Note that the universally truthful auction given by \eqref{eq:mapping} may result rather unnatural mainly because it is not clear with which probability a certain price is charged. However, there are cases in which $\delta$ has some properties for which \eqref{eq:mapping} can be easily decoded. Let $B_1 > B_2 > \ldots > B_T$ be the number of different prices charged by the lottery for some $i$ and $\bi$. If $T$ is finite then $\delta$ is a step function and $X$ a discrete random variable; in this case, one can take the interval $(B_j, B_{j-1}]$ and infer from the proof that the universally truthful auction will charge $B_{j-1}$ with essentially probability $\lambda_i(B_{j-1}, \bi)-\lambda_i(B_{j}, \bi)$. In the case in which $T$ is infinite, with $\delta$ having a derivative $\theta$ which can be integrated back to $\delta$ (e.g., $\lambda$ is continuously differentiable) then $\theta$ is the probability density function of the random variable $X$ meaning that \eqref{eq:mapping} can be rewritten as $\int_x^y \theta(u) du$ so to obtain an infinitesimal weight for the probability of prices in $(x,y]$. 

Finally, let us observe that the proof above holds for col\-lu\-sion-resistance as well when lotteries are singular.

\paragraph{A Matching lower bound for $[1,H]$ domains.}
We can now move back to the case in which bidders have valuations belonging to an interval $[1,H]$ with the aim of understanding how relaxing collusion-resistance to truthfulness affects the approximation guarantee to $\opt$ achievable also in the case of continuous domains. By using a lower bound technique developed in \cite{GHKSW06,GolHar05} for universally truthful auctions and by relying on the bijection between these auctions and truthful lotteries just proved, we can show that it is not possible to achieve an approximation guarantee better than that given in Theorem \ref{thm:cont:UB}, even by relaxing collusion-resistance to truthfulness.

\begin{theorem}\label{thm:cont:lowerbound}
For truthful lotteries and bidders bidding from a domain $[1,H]$, \opt\ cannot be approximated better than $\ln(H)+1$.
\end{theorem}
\begin{proof}
We show that for any truthful lottery $L=(\lambda, p)$ there exists a bid vector $\b \in [1,H]^n$ such that the revenue collected by $L$ on input $\b$ is at most $\opt(\b)/(\ln(H)+1)$. 

To prove this, we analyze the behavior of $L$ on a bid vector $\b$ chosen from a carefully designed probability distribution. The outcome of the lottery is a random variable depending on the randomness in $\b$ and $\lambda$. We prove that the expected revenue of $L$ is at most a $\ln(H)+1$ fraction of the expected optimal revenue. Then, by definition of expectation, 
there must exist a bid vector for which the claim holds.

Because of Theorem \ref{thm:equiv}, $L$ is equivalent to a universally truthful auction $\cal A$. The latter is characterized in terms of a so-called bid-independent auction (see Definition 2.4 in \cite{GHKSW06}). In a bid-independent auction a price is computed as a (possibly randomized) function of $\bi$ only and offered to the bidder $i$. The bidder wins if the price is bounded from above by $b_i$. In the rest of the proof we will argue about the bid-independent auction $\cal A$ defined upon lottery $L$.

Consider the bid vector $\b$ in which each $b_i$ is i.i.d. generated from the distribution in which $\prob{b_i>y}=1/y$ for $y \in [1, H]$ and $\prob{b_i = H}=1/H$. Let $P_i(\bi)$ be the price charged by $\cal A$ to agent $i$ when the bidder declares $b_i$ and the other agents declare $\bi$. Note that, since $\cal A$ is bid-independent then $P_i(\bi)$ is a random variable depending on the randomness of $\cal A$ and $\bi$ only. Let $R_i$ be the expected revenue from bidder $i$ which is $0$ if $b_i < P_i(\bi)$ and $P_i(\bi)$ otherwise.  For $p \geq 0$, $\expec{R_i\left|{P_i(\bi)}\right.=p}=p \cdot \prob{b_i > p \left|{P_i(\bi)}\right.=p} = p \cdot \prob{b_i > p} \leq 1$, where the equality in the penultimate step follows since $b_i$ is independent of ${P_i(\bi)}$. Therefore, 
{\allowdisplaybreaks
\begin{align*}
\expec{R_i}&=\sum_p \expec{R_i \left|{P_i(\bi)}\right.=p}\cdot \prob{{P_i(\bi)}=p} \leq \sum_p \prob{{P_i(\bi)}=p} \leq 1.
\end{align*}
}
We can then conclude that the expected revenue of $\cal A$ on the randomly generated bid vector $\b$ is at most $n$. On the other hand, $\expec{\opt(\b)}=n \expec{b_i} = n (\ln(H)+1)$.
\end{proof}

\section{Conclusions}\label{sec:conclusions}
We consider incentive-compatible auctions in which the auctioneer sells lottery tickets for winning the good being auctioned. We aim at the design of auctions maximizing the revenue; the benchmark of interest is, in this setting, the optimal revenue, defined as the sum of bidders' valuations. Although well motivated, rather technical benchmarks, such as ${\cal F}^{(2)}$ defined above, are instead used in related literature. We study these auctions in the context of digital goods and prove the equivalence, in terms of approximation guarantee to the optimal revenue, between collusion-resistant and truthful auctions (cf. Table~\ref{table}). The former is a much stronger requirement than the latter: it requires the auctions to be resistant to coalitions of cheating bidders exchanging side payments. This equivalence is proved to be true in three different scenarios of bidders' domains: binary, finite and interval $[1,H]$. In the first two cases, the lower bound is proved via a new technique to bound the revenue of truthful lotteries which we regard as our main technical contribution. In the case of bids coming from $[1,H]$, the lower bound is proved by showing the equivalence between universally truthful auctions and truthful lotteries and then applying a known technique due to \cite{GHKSW06}. The proof of equivalence shows that the concept of incentive-compatible lotteries is rather useful: lotteries can be much more natural to imagine than universally truthful auctions.

A number of questions are raised by our results. For example, it would be interesting to evaluate whether the feasibility of using the optimal benchmark is due to the assumptions on the domains, or rather to the expressiveness of lotteries. We believe that a study of lotteries in settings different from that of digital goods can shed light on this important question. Notice, however, that moving from the digital good setting may imply that we have to give up collusion-resistance in order to get any reasonable performance. Indeed, consider a 1-item auction with 2 possible bid values $\{L,H\}$, $L<H$; Suppose that $\lambda_i$ had to depend only on bid $b_i$ and not $\bi$. Since we have only one item available to sell, we need $\sum_{j=1}^n \lambda_j(H)\leq 1$, from which it follows that some $j$ satisfies $\lambda_j(H) \leq 1/n$. Suppose then that all other bidders have value $L$, so that $\opt=H$. All the other bidders would pay at most $L$ while $j$ pays at most $H/n$, so for $H \gg nL$ this fails to approximate \opt\ within any constant.

\bigskip \noindent {\bf Acknowledgments.}
We wish to 
thank Patrick Briest for proposing this study and for the invaluable discussions we had in the starting phases of this project. We are also indebted to him for 
the ideas contained in Section 
\ref{sec:bounded:UB}. We also thank 
Orestis Telelis for his comments on an early draft of this work.

\bibliographystyle{abbrv}

\end{document}